\documentclass[final,3p,times]{interact}
\usepackage{epstopdf}
\usepackage{subfigure}
\usepackage{algorithmic}
\usepackage{algorithm}
\usepackage{arydshln}
\usepackage{color}
\usepackage{caption}
\usepackage{geometry}
\usepackage[colorlinks]{hyperref}
\hypersetup{allcolors=blue}
\usepackage[numbers,sort&compress]{natbib}
\bibpunct[, ]{[}{]}{,}{n}{,}{,}

\theoremstyle{plain}
\newtheorem{theorem}{Theorem}[section]

\newtheorem{proposition}[theorem]{Proposition}
\theoremstyle{definition}

\theoremstyle{remark}

\begin{document}
\title{Model-based Computed Tomography Image Estimation: Partitioning Approach}
\author{
\name{Fekadu~L. Bayisa\thanks{CONTACT Fekadu~L. Bayisa. Email: fekadu.bayisa@umu.se} and Jun Yu}
\affil{Department of Mathematics and Mathematical Statistics,  Ume{\aa} University, Ume{\aa}, Sweden}}
\maketitle
\begin{abstract}
There is a growing interest to get a fully MR based radiotherapy. The most important development needed is to obtain improved bone tissue estimation. The existing  model-based methods perform poorly on bone tissues. This paper was aimed at obtaining improved  bone tissue estimation. Skew-Gaussian mixture model and Gaussian mixture model were proposed to investigate CT image estimation from MR images by partitioning the data into two major tissue types. The performance of the proposed models was evaluated using leave-one-out cross-validation method on real data. In comparison with the existing model-based approaches, the model-based partitioning approach outperformed in bone tissue estimation, especially in dense bone tissue estimation. 
\end{abstract}
\section{Introduction}
Magnetic resonance (MR) imaging and computed tomography (CT) are the most widely used  diagnostic imaging technologies in medicine. They are used to obtain more detailed cross-sectional images of human body. CT uses ionizing radiation to record a pattern of radiodensities to obtain  cross-sectional images.  The ionizing radiations are attenuated as they pass through the tissues of patients. The amount of attenuations depends on the tissue types. The differences in attenuation between adjacent tissues create contrast on CT images. Tissues with  higher (or lower) attenuation appear brighter (or darker) on grayscale CT images. As a result, air, soft, and bone tissues appear as darkest, darker, and white on  grayscale CT images. Therefore, CT image is excellent for identifying and assessing the structures of bone tissues. On the other hand, exposing a patient to ionizing radiation in CT imaging may have a risk for radiation-related cancer. 

MR imaging is remarkably different from CT. It does not depend on ionizing radiations. MR imaging relies on the absorption and
emission of radio waves from tissue protons exposed to a strong magnetic field. Thus, MR imaging is safer than CT imaging. The relative MR signal intensity differences between adjacent anatomic structures determine tissue contrast on MR images. In comparison with CT images, MR images show much better soft tissue contrast  and noticeably improves the delineation of tumors. However, MR images are poor in depicting bone tissues. The reason is that bone, air, and rapidly flowing blood appear black on grayscale MR images.

A new innovation MR-only based radiotherapy enhances tissue contouring and precision in soft tissue therapy setup. It also improves biological information at treatment planning and avoids registration errors, which are errors due to  transformation of different images of the same scene into one coordinate system, between CT and MR images. Moreover, MR-only based radiotherapy is a cost effective approach as it reduces redundant imaging. However, co-registered CT and MR image complement each other due to  better bone tissue imaging of CT \citep{ABO12, ABO4, ABO6}. MR images are not directly applicable for attenuation correction.  Attenuation refers to the loss of  signals due to absorption or scattering out of the signals in the body.  CT images are vital for attenuation correction in positron emission tomography (PET) imaging. This is due to the direct relation between CT image intensities and PET attenuation coefficients. However, CT scanner is not available in recently combined PET/MR imaging scanner. Therefore, MR-only based radiotherapy and the combined PET/MR imaging scanner can be successful if one can obtain a reliable CT image estimation. As a result, we need to develop a reliable CT image estimation method  from MR images.

Huynh et al.  \cite{ABO9} used a learning-based method to estimate CT image from MR image.  A patch of CT image is estimated directly from a given MR image patch using structured random forest. The robustness of the estimation has been evaluated using a new ensemble model. Nie et al. \cite{ABO18} proposed a 3D deep learning-based method for patch-wise estimation of CT images from MR images. The neural network generates structured output and it preserves the neighborhood information in the estimated CT image. Arabi et al. \cite{ABO1} suggested a two-step atlas-based algorithm to estimate CT image from MR image sequences. The estimation is mainly concerned with pinpointing of bone tissues.

Johansson et al. \cite{ABO10} used a Gaussian mixture model (GMM) to obtain CT substitute from MR images without taking spatial dependence between neighboring voxels into account. Johansson et al. \cite{ABO11}   investigated the uncertainty associated to the voxel-wise estimation of CT images. By considering spatial dependence between neighboring voxels, Kuljus et al. \cite{ABO14} extended the work of Johansson et al. \cite{ABO10} by using hidden Markov model (HMM) and Markov random field model (MRF).  Kuljus et al. \cite{ABO14} compared the estimation quality of GMM, HMM, and MRF. In terms of mean absolute error, HMM outperformed the other models and it was computationally robust than MRF. However, it had a weaker estimation quality on dense bone tissues. Even though MRF had superior performance on bone tissue estimation, it was computationally expensive. 

The main aim of this article is to further investigate the voxel-wise estimation of CT images from  MR images by partitioning the data into non-bone and bone tissues.  According to  Johansson et al. \cite{ABO11} and Kuljus et al. \cite{ABO14}, the  estimation of CT image was poor on air and bone tissues. It is this result that motivated the partitioning of the data into non-bone and bone tissues in order to further explore the estimation of CT images. Even though there is no clear-cut CT image intensity boundary between these tissue types, Waterstram-Rich and Gilmore \cite{ABO21} and Washington and Leaver \cite{ABO20} provide informative threshold delimiting these tissues.  Waterstram-Rich and Gilmore  \cite{ABO21} used 150 Hounsfield units (HU) as a lower limit for bone tissues.  On the other hand, Washington and Leaver \cite{ABO20} utilised 200HU as an approximate delimiting value of the tissues.

The partitioning of the data may introduce skewness. Consequently, there is a need to relax the normality assumption used in \citep{ABO10, ABO11, ABO14}.  Azzalini \cite{ABO2} proposed a univariate skew-normal model that relaxed the normality assumption by incorporating a skewness parameter in the distributional assumption.  Azzalini and Dalla  Valle \cite{ABO3} extended the univariate skew-normal to a multivariate skew-normal. A multivariate skew-normal is a tractable extension of a multivariate normal distribution with extra parameter to regulate skewness.  Lin et al. \cite{ABO17} introduced a univariate skew-normal mixture model in order to deal with population heterogeneity and skewness.  Lin \cite{ABO16} extended the univariate skew-normal mixture model to a multivariate skew-normal mixture model which is an alternative to the most widely used multivariate Gaussian mixture model.

Kahrari et al. \cite{ABO22} developed a multivariate skew-normal-Cauchy distribution and represented it as a shape mixture of the multivariate skew-normal distribution. Kahrari et al. \cite{ABO23} modified the multivariate skew-normal-Cauchy distribution and the modified version becomes a shape mixture of a special case of the fundamental skew-normal distribution developed by Arellano-Valle and Genton \cite{ABO24} with a univariate half-normal mixing distribution.  The class of scale mixtures of skew-normal-Cauchy distributions has been represented as a shape mixture of the class of scale mixtures of skew-normal distributions with a univariate half-normal mixing distribution \cite{ABO23}. Jamalizadeh and Lin \cite{ABO25} presented the scale-shape mixtures of skew-normal distributions for modeling asymmetric data. Arellano-Valle et al. \cite{ABO29} established a flexible class of multivariate distributions obtained by both scale and shape mixtures of multivariate skew-normal distributions. Based on the skew-t-normal distribution \cite{ABO27}, Tamandi et al. \cite{ABO26} introduced the shape mixtures of the skew-t-normal distributions that contain one additional shape parameter to regulate skewness and kurtosis. The shape mixtures of the skew-t-normal distributions are a flexible extension of the skew-t-normal distribution. Lin et al. \cite{ABO28} developed a multivariate extension of the skew-t-normal distribution that is obtained as a scale mixture of the multivariate skew-normal distribution introduced by Azzalini and Dalla Valle \cite{ABO3}. Based on the multivariate skew-t-normal distribution, Lin et al. \cite{ABO28} introduced a robust probabilistic mixture model which is composed of a weighted sum of a finite number of different multivariate skew-t-normal densities. The flexible mixture model based on the multivariate skew-t-normal distribution includes mixtures of normal, t and skew-normal distributions as special cases. Cabral et al. \cite{ABO5} proposed mixture models which consist of members of skew-normal independent distributions (the skew-normal, the skew-t, the skew-slash  and the skew-contaminated normal) and the mixture models are developed using the multivariate skew-normal distribution in \citep{ABO3}.
	
The most common approach to estimate the parameters of these skew-mixture models is the EM algorithm \cite{ABO7}. However, its M-step for the recent skew-mixture models is computationally intractable. Alternatively, we use an EM-type algorithm to estimate these skew-mixture models. That is, we  make further assumptions at the M-step of the EM algorithm. For instance,  expectation conditional maximization (ECM) algorithm \cite{ABO30}, which replaces the M-step of EM with several computationally simple conditional maximization steps, can be enough to estimate the parameters of some mixture models.  In mixtures of multivariate skew-t-normal distributions and shape mixtures of skew-t-normal distributions, expectation conditional maximization either (ECME) algorithm \cite {ABO31} was exploited to estimate its parameters by replacing some conditional maximization-steps of ECM with the conditional maximization likelihood step that maximizes the correspondingly constrained actual-likelihood function. Cabral et al. \cite{ABO5} developed an EM-type algorithm that removed some obstacles (for instance, Monte Carlo integration) during parameter estimation process in mixture models of skew-normal independent distributions. In this article, we used a mixture model based on the multivariate skew-normal distribution in \citep{ABO3} and developed EM-algorithm for its parameter estimation. That is, further assumption is not required at M-step of EM algorithm to estimate the parameters of the skew-Gaussian mixture model.

In this work, we use skew-Gaussian mixture model (SGMM) and Gaussian mixture model (GMM) to further explore the estimation of CT images by partitioning the data into two major tissue types: non-bone and bone tissues. Non-bone tissues consist of subclasses such as white matter, blood, water, fat,  gray matter, air, etc while bone tissues  consist of subclasses such as cranium, mandible, frontal bone, nasal bone, orbital bones, cortical bone, cancellous bone etc. These facts motivated us to apply mixture models on these tissue types. SGMM involves a weighted sum of the joint skew-normal distributions of a CT image intensity and its corresponding intensities of MR images. The number of skew-normal distributions in the mixture depends on the number of underlying tissue types  or clusters.  Latent variables that represent the underlying tissue types are utilized during   parameter estimation process of the model through incomplete data assumption in EM-algorithm framework \citep{ABO7}. Voxel-wise point estimator of CT image was obtained as a weighted sum of the conditional expected value of a CT image intensity given its corresponding intensities of MR images and the underlying tissue type. The probability that an underlying tissue type is determined based on the intensities of MR images was used as a weight of the conditional expected value. 

In summary, this study is concerned with comparing the CT image estimation performance of the partitioning approach with HMM, MRF, and GMM on bone tissues. The models HMM , MRF, and GMM are trained on the full data (data that are not partitioned into non-bone and bone tissues). We are also interested to compare the predictive quality of the partitioning approaches, SGMM and GMM* (GMM applied to each partition), on bone tissues. 

This article is organized as follows. The second section describes data acquisition and demonstrates statistical methods.  The third section presents the results obtained and the final section discusses the implication of the results. 
\section{Statistical methodology}
In this section, we describe the data, formulate SGMM, and develop parameter estimation method. We also demonstrate CT image estimation method and present evaluation method of the estimation. For the remaining models (GMM, HMM and MRF), we refer to  Kuljus et al. \citep{ABO14} and the references therein.
\subsection{Data acquisition}
CT and MR images were obtained from head of five patients. Four MR images were acquired from each patient using two dual echo ultrashort echo-time sequences with flip angles of 10 degrees and 30 degrees. The ultrashort echo-time sequences sampled a first echo (free induction decay) and a second echo (gradient echo) from the same excitation with an echo time of 0.07 and 3.76 milliseconds.  MR image of a patient was reconstructed to $192\times192\times192$ matrix. An entry in the matrix represents a signal intensity corresponding to a three dimensional tissue (voxel) with size 1.33  $\times$1.33 $\times$1.33 mm$^{3}$. One CT image of a patient was acquired using gradient echo Lightspeed with 2.5mm slice thickness. The acquired CT image was reconstructed with an in-plane resolution of 0.78  $\times$ 0.78 mm$^{2}$. One binary mask (an image with voxel value 1 (or 0) representing the region of interest (or the surrounding air)) was also developed to demarcate the head of a patient from its surrounding air. The main use of the binary mask is to exclude the surrounding air from the acquired CT and MR images. For each patient, the binary mask, the CT image, and the four MR images were co-registered and resampled to the same resolution (voxel-to-voxel correspondence and set to the same voxel dimension) using linear interpolation. For further technical details, we refer to Johansson et al. \citep{ABO10}. Voxel values of the CT image, the binary mask, and the four MR images were organised into six columns to obtain data for a patient. The organised data of each patient were column stacked and the surrounding air removed to obtain data for model fitting. Figure \ref{Or1} shows a slice data for a given  patient. 

\begin{figure}[H]
	\centering
	\captionsetup{justification=centering}
	\includegraphics[height=8cm, width=12cm]{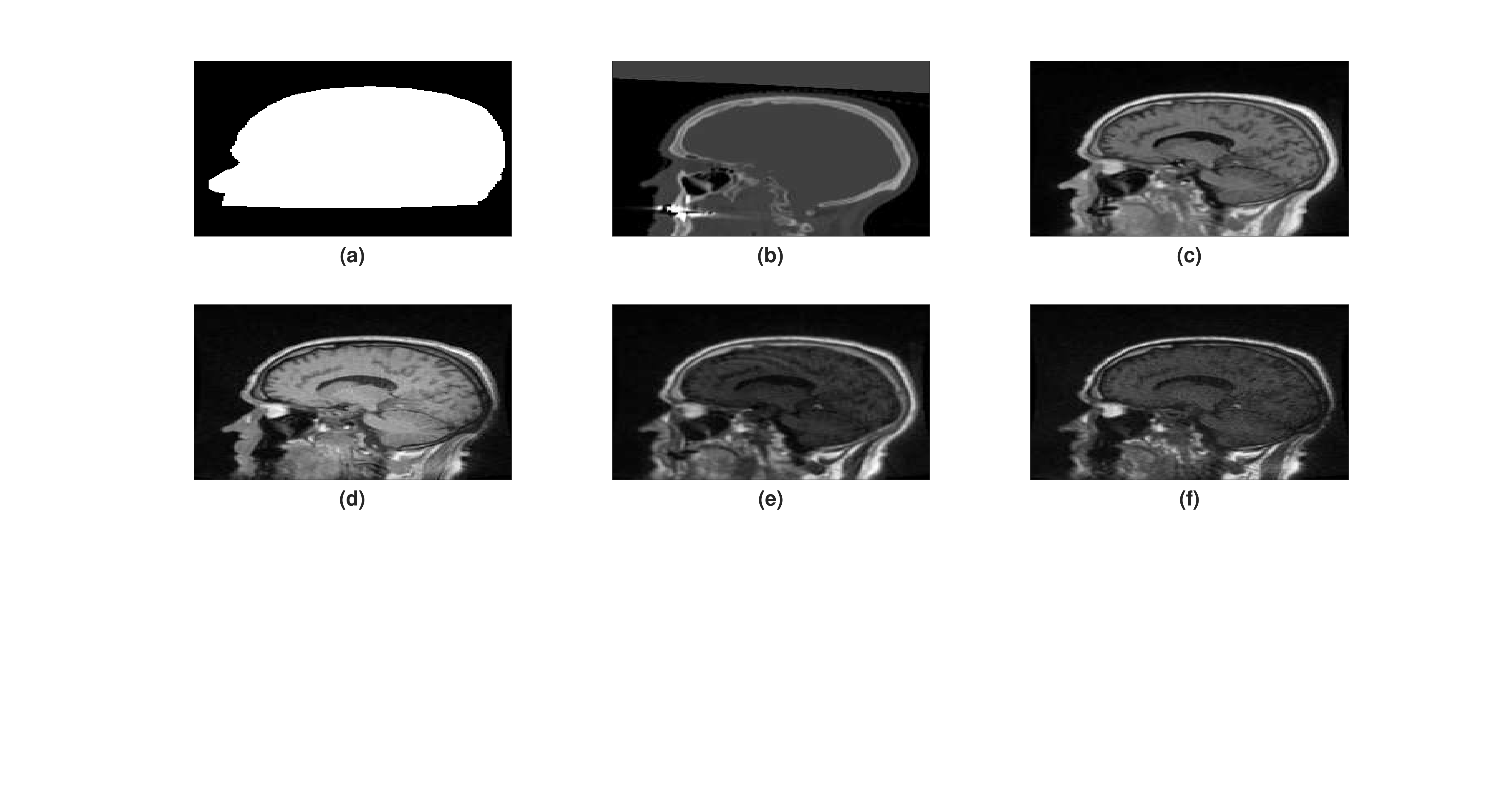}
\caption{Binary mask ((a)), CT image ((b)) and MR images ((c)-(f))}
	\label{Or1}
\end{figure}
\subsection{Data Partitioning}
This subsection describes data partitioning during model training. It also demonstrates how  MR images of new patients are utilized during CT image prediction.
\subsubsection{Data partition: Model training}	
CT image intensity threshold was utilized to partition the data into two major tissue types. Using 150HU CT image intensity as a limit, 50HU (is selected to take the delimiting value provided by Washington and Leaver \citep{ABO20}) overlap was allowed during parameter estimation process. The overlap of the major tissue types was motivated in order to minimize the effect of fuzzy boundary of the tissue types. Accordingly, CT image intensities in (-1024HU, 200HU) and (100HU, 3071HU] were assumed to represent non-bone and bone tissues. The minimum number of voxels for non-bone and bone tissues are 6214160 and 1292068. 
\subsubsection{MR images partition: CT image estimation}
We are interested to predict CT images from MR images of new patients. Since we only have MR images of the new patients, there is a need to partition the MR images of the new patients into non-bone and bone tissues.  There is poor bone tissue information on MR images and following that we estimate CT images for the new patients using a model trained on the full data, which is the data not partitioned into non-bone and bone tissues. The CT image intensity threshold and the estimated CT images are used to obtain MR data corresponding to the two major tissue types (non-bone and bone tissues). The trained models on the major tissue types are utilized to obtain the desired CT images of the new patients. 
\subsection{Statistical model: mixture of multivariate skew-normal model}
Let $Y_{i1}$ and $\mathbf{Y}_{i2} = \left(Y_{i2}, Y_{i3} \cdots, Y_{id}\right)'$ represent voxel $i$ of CT image and its corresponding MR images. In our real data, d=5.  A $d$-dimensional random vector $\mathbf{Y}_{i}=\left(Y_{i1}, \mathbf{Y}_{i2}'\right)'$ is assumed to follow a multivariate skew-normal distribution $\mathcal{SN}\left(\mathbf{y}_{i}|\boldsymbol\eta, \boldsymbol\Sigma, \boldsymbol\lambda\right)$ with a $d$-dimensional location parameter vector $\boldsymbol\eta$, a $d \times d$-dimensional positive definite dispersion matrix $\boldsymbol\Sigma$, and a $d$-dimensional skewness parameter vector $\boldsymbol\lambda$. Its density can be given by
\begin{eqnarray}\label{Oromo1}
f\left(\mathbf{y}_{i}|\boldsymbol\eta, \boldsymbol\Sigma, \boldsymbol\lambda\right) = 2\mathcal{N}\left(\mathbf{y}_{i}|\boldsymbol\eta, \boldsymbol\Sigma\right)\Phi\left(\boldsymbol\lambda'\boldsymbol\Sigma^{-1/2}\left(\mathbf{y}_{i}-\boldsymbol\eta \right)\right),  
\end{eqnarray}
where $\boldsymbol\Sigma^{-1/2}\boldsymbol\Sigma^{-1/2}=\boldsymbol\Sigma^{-1}$, $\Phi\left(\cdot\right) $ is a univariate standard normal distribution function,  $i = 1, 2, 3, \cdots, n$, and $n$ is the number of voxels. 
According to Lachos et al. \cite{ABO15}, the stochastic representation of  $\mathbf{Y}_{i}$ may be given by 
\begin{eqnarray}\label{Oromia6}
\mathbf{Y}_{i} = \boldsymbol\eta + \boldsymbol\Sigma^{\frac{1}{2}}\boldsymbol\delta U_{i} + \boldsymbol\Sigma^{\frac{1}{2}}\left(\mathbf{1}-\boldsymbol\delta\boldsymbol\delta'\right)^{\frac{1}{2}}\mathbf{V}_{i},
\end{eqnarray}
where
\begin{equation}\label{mekdes}
U_{i}\sim\mathcal{HN}\left(u_{i}|0, 1, \left(0, \infty\right)\right),  \quad\mathbf{V}_{i}\sim\mathcal{N}\left(\mathbf{v}_{i}|\mathbf{0}, \mathbf{1}\right),\quad \text{and} \quad \boldsymbol\delta =\frac{\boldsymbol\lambda}{\sqrt{1+\boldsymbol\lambda'\boldsymbol\lambda}}.
\end{equation}
In equation \eqref{Oromia6}, $\mathbf{V}_{i}$ and $U_{i}$ are assumed to be independent. The notations $\mathbf{0}$, $\mathbf{1}$, and $\mathcal{HN}\left(u_{i}|0, 1, \left(0, \infty\right)\right)$ in equation \eqref{mekdes} represent  $d$-dimensional zero vector,  $d\times  d$-dimensional identity matrix,  and  half-normal distribution, respectively.   Using equation \eqref{Oromia6}, a hierarchical model can be given by
\begin{eqnarray}\label{Oromia7}
\begin{split}
\mathbf{Y}_{i}|U_{i} = u_{i}& \sim \mathcal{N}\left(\mathbf{y}_{i}|\boldsymbol\eta + u_{i}\boldsymbol\xi , \boldsymbol\Omega \right),\\
U_{i}&\sim\mathcal{HN}\left(u_{i}|0, 1, \left(0, \infty\right)\right),
\end{split}
\end{eqnarray}
where
\begin{equation*}
\boldsymbol\delta =\frac{\boldsymbol\lambda}{\sqrt{1+\boldsymbol\lambda'\boldsymbol\lambda}},\quad\boldsymbol\xi=\boldsymbol\Sigma^{\frac{1}{2}}\boldsymbol\delta,\quad \text{and}\quad\boldsymbol\Omega= \boldsymbol\Sigma^{\frac{1}{2}}\left(\mathbf{1}-\boldsymbol\delta\boldsymbol\delta'\right)\boldsymbol\Sigma^{\frac{1}{2}}=\boldsymbol\Sigma-\boldsymbol\xi\boldsymbol\xi'. \notag
\end{equation*} 
Let $Z_{i}$ be a categorical random variable representing the underlying tissue types at voxel $i$. Define an indicator variable:
\begin{eqnarray*}
	Z_{ik}=1_{\left(Z_{i}=s\right)}=\begin{cases}
		1, & \text{if } k=s,\\
		0, & \text{otherwise},
	\end{cases}
\end{eqnarray*}
where $k = 1, 2, \cdots, K$. The definition implies that $P\left(Z_{ik} = 1\right)= P\left(Z_{i}=k\right)$. Let $P\left(Z_{i}=k\right)$= $\pi_{k}$, which represents the weight that the $i^{th}$ observation belongs to a tissue class $k$. To incorporate tissue heterogeneity into the statistical modeling, $\mathbf{Y}_{i}|Z_{ik} = 1$ is assumed to follow  a multivariate skew-normal  distribution $\mathcal{SN}\left(\mathbf{y}_{i}|\boldsymbol\eta_{k}, \boldsymbol\Sigma_{k}, \boldsymbol\lambda_{k}\right)$. This means that $\mathbf{Y}_{i}$ follows a mixture of  multivariate skew-normal  distributions. Its density may be given by 
\begin{align*}
f\left(\mathbf{y}_{i}|\boldsymbol\Psi\right) = \sum_{k=1}^{K}\pi_{k}\mathcal{SN}\left(\mathbf{y}_{i}|\boldsymbol\eta_{k}, \boldsymbol\Sigma_{k}, \boldsymbol\lambda_{k}\right),
\end{align*}
where	
\begin{equation*}	
\pi_{k} \ge 0,\hspace{0.01in} \sum_{k=1}^{K}\pi_{k}=1,
\hspace{0.05in}\boldsymbol\psi_{k}=\left\lbrace \pi_{k},\boldsymbol\eta_{k}, \boldsymbol\Sigma_{k}, \boldsymbol\lambda_{k}\right\rbrace ,\hspace{0.05in} \boldsymbol\Psi = \left\lbrace \boldsymbol\psi_{1}, \boldsymbol\psi_{2}, \cdots, \boldsymbol\psi_{K}\right\rbrace,\hspace{0.05in}k= 1, 2, \cdots, K. 
\end{equation*}
The unknown  parameters in $\boldsymbol\Psi$ are estimated from the independent observations $\mathbf{y}_{i}$. 
\subsection{Parameter estimation method }
\noindent The log-likelihood function of the data $\mathbf{y} = \left(\mathbf{y}_{1}, \mathbf{y}_{2}, \cdots, \mathbf{y}_{n} \right)'$ can be given by  
\begin{eqnarray*}
	\log f\left(\mathbf{y}|\boldsymbol\Psi\right)=\sum_{i=1}^{n} \log \left\lbrace \sum_{k=1}^{K}\pi_{k}\mathcal{SN}\left(\mathbf{y}_{i}|\boldsymbol\eta_{k}, \boldsymbol\Sigma_{k}, \boldsymbol\lambda_{k}\right)\right\rbrace.
\end{eqnarray*}
In general, there is no explicit analytical solution for 
$\arg\hspace{-0.03in}\max\limits_{\boldsymbol\Psi} \log f\left(\mathbf{y}|\boldsymbol\Psi\right)$. However, iterative maximizing  procedure under the idea of incomplete data via EM-algorithm can be used to obtain an optimal estimate of the parameters. Let $\mathbf{Z}_{i}$ be a $K$-dimensional column vector of  $Z_{ik}$.  Its realization is a $K$-dimensional vector consisting 1 at only one location and 0 at the remaining locations. The latent random vector $\mathbf{Z}_{i}$ follows a multinomial distribution with one trial and $P\left(Z_{ik} = 1\right)=\pi_{k}$. Using the indicator variable $Z_{ik}$, the hierarchical model \eqref{Oromia7} can be extended to:
\begin{align*}
\begin{split}
\mathbf{Y}_{i}|U_{i} = u_{i}, Z_{ik} = 1&\sim \mathcal{N}\left(\mathbf{y}_{i}|\boldsymbol\eta_{k} + u_{i}\boldsymbol\xi_{k}, \boldsymbol\Omega_{k}\right), \\ 
U_{i}|Z_{ik} = 1&\sim\mathcal{HN}\left(u_{i}|0, 1, \left(0, \infty\right)\right),
\end{split}
\end{align*}
where
\begin{equation*}
\boldsymbol\delta_{k} =\frac{\boldsymbol\lambda_{k}}{\sqrt{1+\boldsymbol\lambda_{k}'\boldsymbol\lambda_{k}}},\quad\boldsymbol\xi_{k}=\boldsymbol\Sigma_{k}^{\frac{1}{2}}\boldsymbol\delta_{k},\quad\text{and}\quad
\boldsymbol\Omega_{k} = \boldsymbol\Sigma_{k}-\boldsymbol\xi_{k}\boldsymbol\xi_{k}'.
\end{equation*}
The observed data $\mathbf{y}$ is assumed to be incomplete data. It is augmented with the latent matrix $\mathbf{z} = \left(\mathbf{z}_{1} , \mathbf{z}_{2}, \cdots, \mathbf{z}_{n}\right)'$ and the latent vector $\mathbf{u} = \left(u_{1}, u_{2},\cdots, u_{n}\right)'$ to form a complete dataset $\left(\mathbf{y}, \mathbf{u}, \mathbf{z}\right)$ in EM-algorithm framework. Assuming that $\left(\mathbf{Y}_{i}, U_{i}, \mathbf{Z}_{i}\right)$ is independent of  $\left(\mathbf{Y}_{j},  U_{j}, \mathbf{Z}_{j}\right)$  for every $i\ne j$, the complete-data log-likelihood is given by
\begin{eqnarray}\label{Oromo8}
\begin{split}
\log f\left(\mathbf{y}, \mathbf{u}, \mathbf{z}|\boldsymbol\Theta\right) &=& \sum_{i=1}^{n}\sum_{k=1}^{K}z_{ik}\left\lbrace -\frac{1}{2}\left(\mathbf{y}_{i}-\boldsymbol\eta_{k}-\boldsymbol\xi_{k}{u}_{i} \right)'\boldsymbol\Omega_{k}^{-1}\left(\mathbf{y}_{i}-\boldsymbol\eta_{k}-\boldsymbol\xi_{k}{u}_{i}\right)\right\rbrace+\\&&\sum_{i=1}^{n}\sum_{k=1}^{K}z_{ik}\left\lbrace\log\pi_{k}-\frac{1}{2}\log|\boldsymbol\Omega_{k}|+ const\right\rbrace,
\end{split}
\end{eqnarray}
where
\begin{equation*}
\boldsymbol\theta_{k}=\left\lbrace \pi_{k},\boldsymbol\eta_{k},\boldsymbol\xi_{k}, \boldsymbol\Omega_{k}\right\rbrace,\quad \boldsymbol\Theta = \left\lbrace \boldsymbol\theta_{1}, \boldsymbol\theta_{2}, \cdots, \boldsymbol\theta_{K}\right\rbrace, \quad k= 1, 2, \cdots, K, 
\end{equation*}
and $const$ is a constant function of parameters.

The E-step of EM-algorithm involves computing the expected value of the complete-data log-likelihood given $\mathbf{Y}$ and the current estimate $\boldsymbol\Theta^{old}$ of $\boldsymbol\Theta$. It may be given by the Q-function:
\begin{align}\label{Oromo9}
\mathcal{Q}\left(\boldsymbol\Theta,\boldsymbol\Theta^{\text{old}}\right) = E\left[\log f\left(\mathbf{Y},\mathbf{U}, \mathbf{Z}|\boldsymbol\Theta\right)|\mathbf{Y},\boldsymbol\Theta^{\text{old}}\right].
\end{align}
Using equation \eqref{Oromo8}, the expected value in equation \eqref{Oromo9} involves computing: 
\begin{eqnarray*}
	E\left[Z_{ik}|\mathbf{Y},\boldsymbol\Theta^{\text{old}}\right],\quad E\left[Z_{ik} U_{i}|\mathbf{Y},\boldsymbol\Theta^{\text{old}}\right],\quad \text{and}\quad E\left[Z_{ik} U^{2}_{i}|\mathbf{Y},\boldsymbol\Theta^{\text{old}}\right].
\end{eqnarray*}
The expected value $E\left[Z_{ik}|\mathbf{Y},\boldsymbol\Theta^{\text{old}}\right]$ can be given by
\begin{eqnarray*}
	E\left[Z_{ik}|\mathbf{Y},\boldsymbol\Theta^{\text{old}}\right]&=& \frac{\pi^{\text{old}}_{k}\mathcal{SN}\left(\mathbf{y}_{i}|\boldsymbol\eta_{k}^{\text{old}}, \boldsymbol\Sigma_{k}^{\text{old}}, \boldsymbol\lambda_{k}^{\text{old}} \right)}{\displaystyle\sum_{j=1}^{K}\pi^{\text{old}}_{j}\mathcal{SN}\left(\mathbf{y}_{i}|\boldsymbol\eta_{j}^{\text{old}}, \boldsymbol\Sigma_{j}^{\text{old}}, \boldsymbol\lambda_{j}^{\text{old}} \right)},\\
	&=&\gamma^{\text{old}}_{ik},
\end{eqnarray*}
where $\gamma_{ik}$ is the responsibility that the component $k$ of the mixture takes for explaining the observation $\mathbf{y}_{i}$. Let $\vartheta_{ik} = E\left[Z_{ik} U_{i}|\mathbf{Y},\boldsymbol\Theta^{\text{old}}\right]$, and $\psi_{ik}= E\left[Z_{ik} U^{2}_{i}|\mathbf{Y},\boldsymbol\Theta^{\text{old}}\right]$.  Then $\vartheta_{ik}$ can be simplified as follows: 
\begin{eqnarray}\label{Oromo11}
\vartheta_{ik}&=&E\left[Z_{ik} E\left[ U_{i}|\mathbf{Y},Z_{ik},\boldsymbol\Theta^{\text{old}}\right] |\mathbf{Y},\boldsymbol\Theta^{\text{old}}\right],\notag\\
&=&E\left[U_{i}|\mathbf{Y},Z_{ik}=1,\boldsymbol\Theta^{\text{old}}\right]E\left[Z_{ik} |\mathbf{Y},\boldsymbol\Theta^{\text{old}}\right],\notag\\
&=&\gamma^{\text{old}}_{ik}E\left[U_{i}|\mathbf{Y},Z_{ik}=1,\boldsymbol\Theta^{\text{old}}\right].
\end{eqnarray}
Using similar procedure,   $\psi_{ik}$ may be given by
\begin{eqnarray}\label{Oromo10}
\psi_{ik}=\gamma^{\text{old}}_{ik}E\left[U^{2}_{i}|\mathbf{Y},Z_{ik}=1,\boldsymbol\Theta^{\text{old}}\right].
\end{eqnarray}
If equations \eqref{Oromo1} and \eqref{Oromia7} are satisfied, then using inverse matrix adjustment formula in \textcolor{blue}{\citep{ABO19}} and matrix determinant lemma in \textcolor{blue}{\citep{ABO8}},
\begin{equation*}
U_{i}|\mathbf{Y}_{i}=\mathbf{y}_{i} \sim \mathcal{TN}\left(u_{i}|\frac{\boldsymbol\xi'\boldsymbol\Omega^{-1}\left( \mathbf{y}_{i}-\boldsymbol\eta\right)}{1+\boldsymbol\xi'\boldsymbol\Omega^{-1}\boldsymbol\xi}, \frac{1}{1+\boldsymbol\xi'\boldsymbol\Omega^{-1}\boldsymbol\xi},\left(0, \infty\right) \right), 
\end{equation*}
where $\mathcal{TN}\left(u| \mu,\sigma^{2},\left(0, \infty\right) \right)$ is a truncated normal with location parameter $\mu$, scale parameter $\sigma$ and support $\left(0, \infty\right)$. Based on the truncated normal probability distribution,
\begin{eqnarray}\label{Oromo12}
E\left[U_{i}|\mathbf{Y}_{i}=\mathbf{y}_{i},U_{i}>0\right] &=& \frac{1}{\beta}\left[\alpha+\frac{\phi\left(\alpha\right) }{ \Phi\left(\alpha\right)}\right] ,\label{Oromo12}\\
var\left[U_{i}|\mathbf{Y}_{i}=\mathbf{y}_{i},U_{i}>0\right]&=&\frac{1}{\beta^{2}}\left[1-\alpha\frac{\phi\left(\alpha\right)}{\Phi\left(\alpha \right)}-\left(\frac{\phi\left(\alpha\right)}{\Phi\left(\alpha\right)} \right)^{2}\right],\notag\\
E\left[U_{i}^{2}|\mathbf{Y}_{i}=\mathbf{y}_{i},U_{i}>0\right] &=& \frac{1}{\beta^{2}}\left[1 +\alpha\frac{\phi\left(\alpha\right)}{\Phi\left(\alpha\right)}+\alpha^{2}\right],\label{Oromo13}
\end{eqnarray}
where 
\begin{equation*}
\alpha = \frac{\boldsymbol\xi'\boldsymbol\Omega^{-1}\left( \mathbf{y}_{i}-\boldsymbol\eta\right)}{\sqrt{1+\boldsymbol\xi'\boldsymbol\Omega^{-1}\boldsymbol\xi}},\quad \beta= \sqrt{1+\boldsymbol\xi'\boldsymbol\Omega^{-1}\boldsymbol\xi},
\end{equation*}
and $\phi\left({\cdot}\right)$ is a univariate standard normal density. The expected values:
\begin{equation*}
E\left[U_{i}|\mathbf{Y},Z_{ik}=1,\boldsymbol\Theta^{\text{old}}\right] \quad  \text{and} \quad  E\left[U^{2}_{i}|\mathbf{Y},Z_{ik}=1,\boldsymbol\Theta^{\text{old}}\right],
\end{equation*}
in equations \eqref{Oromo11} and \eqref{Oromo10} are obtained from equations \eqref{Oromo12} and \eqref{Oromo13} by replacing $\boldsymbol\eta$, $\boldsymbol\xi$, and $\boldsymbol\Omega$ with their corresponding estimates $\boldsymbol\eta_{k}^{old}$, $\boldsymbol\xi_{k}^{old}$, and $\boldsymbol\Omega_{k}^{old}$. 
The M-step of the EM-algorithm is given by 
\begin{eqnarray*}
	\boldsymbol\Theta^{\text{new}}=\arg\hspace{-0.03in}\max\limits_{\boldsymbol\Theta}\mathcal{Q}\left(\boldsymbol\Theta,\boldsymbol\Theta^{\text{old}}\right),
\end{eqnarray*}
and it is available in closed form. Using EM-algorithm, the estimates of the  parameter $\boldsymbol\Theta$ can be updated as shown in Algorithm \ref{Oromo14}.
\begin{algorithm}[H]
	\caption{EM-algorithm for SGMM.}
	\begin{algorithmic}[1]\label{Oromo14}
		\STATE Initial values of the parameters ($m$ = 0):\quad $\pi^{\left(m\right)}_{k}$,\quad $\boldsymbol\Omega_{k}^{\left(m\right)}$,\quad $\boldsymbol\eta_{k}^{\left(m\right)}$, \quad and \quad $\boldsymbol\xi_{k}^{\left(m\right)}$; 
		\STATE E-step: \label{Oromo15}
		\begin{eqnarray*}
	\gamma^{\left(m\right)}_{ik}&=& \frac{\pi^{\left(m\right)}_{k}\mathcal{SN}\left(\mathbf{y}_{i}|\boldsymbol\eta_{k}^{\left(m\right)}, \boldsymbol\Sigma_{k}^{\left(m\right)}, \boldsymbol\lambda_{k}^{\left(m\right)} \right)}{\displaystyle\sum_{j=1}^{K}\pi^{\left(m\right)}_{j}\mathcal{SN}\left(\mathbf{y}_{i}|\boldsymbol\eta_{j}^{\left(m\right)}, \boldsymbol\Sigma_{j}^{\left(m\right)}, \boldsymbol\lambda_{j}^{\left(m\right)}\right)};\\
	\vartheta^{\left(m\right)}_{ik}&=& \frac{\gamma^{\left(m\right)}_{ik}}{\beta_{k}^{\left(m\right)}}\left[\alpha_{ik}^{\left(m\right)}+\frac{\phi\left(\alpha_{ik}^{\left(m\right)}\right) }{ \Phi\left(\alpha_{ik}^{\left(m\right)}\right)}\right]; \\
	\psi^{\left(m\right)}_{ik} &=& \frac{\gamma^{\left(m\right)}_{ik}}{\left[\beta_{k}^{\left(m\right)}\right]^{2}}\left[1 +\alpha_{ik}^{\left(m\right)}\frac{\phi\left(\alpha_{ik}^{\left(m\right)}\right)}{\Phi\left(\alpha_{ik}^{\left(m\right)}\right)}+\left[ \alpha_{ik}^{\left(m\right)}\right] ^{2}\right];
\end{eqnarray*}  
		\STATE M-step:
		\begin{eqnarray*}
			\pi^{\left(m+1\right)}_{k}&=&\frac{1}{n}\displaystyle\sum_{i=1}^{n}\gamma^{\left(m\right)}_{ik};\\
			\boldsymbol\eta^{\left(m+1\right)}_{k}&=&\left[\displaystyle\sum_{i=1}^{n}\gamma^{\left(m\right)}_{ik}\right] ^{-1}\displaystyle\sum_{i=1}^{n}\left[ \gamma^{\left(m\right)}_{ik}\mathbf{y}_{i}-\vartheta^{\left(m\right)}_{ik}\boldsymbol\xi_{k}^{\left(m\right)}\right];\\
			\boldsymbol\xi_{k}^{\left(m+1\right)}&=&\left[\displaystyle\sum_{i=1}^{n}\psi^{\left(m\right)}_{ik}\right]^{-1} \displaystyle\sum_{i=1}^{n}\vartheta^{\left(m\right)}_{ik}\left[\mathbf{y}_{i}-\boldsymbol\eta_{k}^{\left(m+1\right)}\right] ;\\				 
			\boldsymbol\Omega^{\left(m+1\right)}_{k}&=&\left[\displaystyle\sum_{i=1}^{n}\gamma^{\left(m\right)}_{ik}\right]^{-1}\displaystyle\sum_{i=1}^{n}\biggl\{ \gamma^{\left(m\right)}_{ik}\left(\mathbf{y}_{i}-\boldsymbol\eta_{k}^{\left(m+1\right)}\right)\left(\mathbf{y}_{i}-\boldsymbol\eta_{k}^{\left(m+1\right)}\right)'-\\&& \vartheta^{\left(m\right)}_{ik}\biggl[\left( \mathbf{y}_{i}-\boldsymbol\eta_{k}^{\left(m+1\right)}\right) \boldsymbol\xi_{k}^{\left(m+1\right)'}+\boldsymbol\xi_{k}^{\left(m+1\right)}\left( \mathbf{y}_{i}-\boldsymbol\eta_{k}^{\left(m+1\right)}\right)' \biggr] +\\&& \psi^{\left(m\right)}_{ik}\boldsymbol\xi_{k}^{\left(m+1\right)}\boldsymbol\xi_{k}^{\left(m+1\right)'} \biggr\};
		\end{eqnarray*}
		\STATE If \emph{stopping criterion} is achieved, \emph{stop}. If not, assign $m+1$ to $m$ and go to step \ref{Oromo15};
	\end{algorithmic}
\end{algorithm}
At $m$th iteration of the E-step, we need to compute the following expressions: 
\begin{eqnarray*}
	\boldsymbol\Sigma^{\left(m\right)}_{k}&=& \boldsymbol\Omega_{k}^{\left(m\right)}+\boldsymbol\xi^{\left(m\right)}_{k}\boldsymbol\xi^{\left(m\right)'}_{k}; \\\boldsymbol\lambda^{\left(m\right)}_{k}&=&\frac{\left[ \boldsymbol\Sigma^{\left(m\right)}_{k}\right]^{-1/2}\boldsymbol\xi_{k}^{\left(m\right)}}{\sqrt{1-\boldsymbol\xi_{k}^{\left(m\right)'} \left[ \boldsymbol\Sigma^{\left(m\right)}_{k}\right]^{-1}  \boldsymbol\xi_{k}^{\left(m\right)}}};\\
	\alpha_{ik}^{\left(m\right)}&=&\frac{\boldsymbol\xi_{k}^{\left(m\right)'}\left[\boldsymbol\Omega^{\left(m\right)}_{k}\right]^{-1}\left( \mathbf{y}_{i}-\boldsymbol\eta_{k}^{\left(m\right) }\right)}{ \beta_{k}^{\left(m\right)}};\\
	\beta_{k}^{\left(m\right)} &=& \sqrt{1+\boldsymbol\xi_{k}^{\left(m\right)'}\left[ \boldsymbol\Omega_{k}^{\left(m\right)}\right] ^{-1}\boldsymbol\xi_{k}^{\left(m\right)}}.
\end{eqnarray*}

In general, EM-algorithm converges to a local optimum. As a result, different initial values for the parameters are utilized during the estimation process to select the optimal estimates. A clustering method $K$-means has been employed to initialize the location parameters, the mixing coefficients, and the dispersion matrices. This clustering method is used to partition the observations into clusters in which each observation belongs to the cluster with the nearest mean. In this study, the clusters may represent tissue types or mixture of tissue types such as white matter, blood, water, fat,  gray matter, air, cortical bone, cancellous bone, etc. We  randomly initialise the remaining  parameters of SGMM. Log-likelihood value can not be computed analytically for MRF \citep{ABO14} and therefore, we can not use it for selecting the optimal estimates. Following that we used mean squared error as criterion for selecting optimal estimates in SGMM in order to utilize the same criterion for the models used in this work. Two steps are employed during parameter estimation process. For a given number of tissue types, we estimated the parameters of the models and  repeated this step to select the optimal parameter estimates using mean squared error. We selected the number of classes or tissue types through cross-validation using mean squared error.  The stopping criterion for the convergence of the parameter estimation process is:
\begin{equation*}
 \displaystyle\max_{ik}|\gamma_{ik}^{\left(m+1\right)}-\gamma_{ik}^{\left(m\right)}|,
\end{equation*}
with an upper limit $5\times 10^{-5}$, where $m$ denotes the iteration number of EM-algorithm.
\subsection{Estimation of CT images}
\noindent 
Let a $d$-dimensional vectors $\boldsymbol\eta$, $\boldsymbol\nu =\boldsymbol\Sigma^{-1/2}\boldsymbol\lambda$,  and a $d\times d$ dispersion matrix $\boldsymbol\Sigma$ be partitioned as follows: 
\begin{equation*}
\mathbf{Y}_{i}=
\begin{array}{c}
\end{array}{\left[
	\begin{array}{c}
	Y_{i1}\\
	\hdashline
	\mathbf{Y}_{i2} 
	\end{array}
	\right]},
\quad
\boldsymbol\eta=
\begin{array}{c}
\end{array}{\left[
	\begin{array}{c}
	\eta_{1}\\
	\hdashline
	\boldsymbol\eta_{2} 
	\end{array}
	\right]},
\quad\boldsymbol\nu=
\begin{array}{c}
\end{array}{\left[
	\begin{array}{c}
	\nu_{1}\\
	\hdashline
	\boldsymbol\nu_{2} 
	\end{array}
	\right]}, \quad\text{and}
\hspace{0.05in}\boldsymbol\Sigma=
\begin{array}{c}
\end{array}{
	\left[
	\begin{array}{c;{2pt/2pt}c}
	\Sigma_{11} & \boldsymbol\Sigma_{12} \\
	\hdashline
	\boldsymbol\Sigma_{21} & \boldsymbol\Sigma_{22}
	\end{array}
	\right]}.
\end{equation*}
The dimension of  $Y_{i1}$, $\eta_{1}$ , $\nu_{1}$, and $\Sigma_{11}$ is a $1\times 1$. The random variable  $Y_{i1}$ represents the $i^{th}$ voxel in CT image, and the random vector $\mathbf{Y}_{i2}$ denotes  the corresponding voxel in MR images.
If we assume that $\mathbf{Y}_{i}$ follows SGMM: 
\begin{equation}\label{ABOWBO1}
\mathbf{Y}_{i}\sim \mathcal{SN}\left(\mathbf{y}_{i}|\boldsymbol\eta, \boldsymbol\Sigma, \boldsymbol\lambda\right),
\end{equation}
then 
\begin{equation*}
\mathbf{Y}_{i2}\sim\mathcal{SN}\left(\mathbf{y}_{i2}|\boldsymbol\eta_{2}, \boldsymbol\Sigma_{22}, \boldsymbol\tau\right),
\end{equation*} 
where  
\begin{equation*}
\boldsymbol\tau=\boldsymbol\Sigma_{22}^{1/2}\left( \frac{\boldsymbol\Sigma_{22}^{-1}\boldsymbol\Sigma_{21}\nu_{1}+\boldsymbol\nu_{2}}{\sqrt{1+\nu_{1}\Sigma^{c}_{11}\nu_{1}}}\right)\quad\text{with}\quad\Sigma^{c}_{11}=\Sigma_{11}-\boldsymbol\Sigma_{12}\boldsymbol\Sigma^{-1}_{22}\boldsymbol\Sigma_{21}.
\end{equation*}
According to  Khounsiavash et al. \citep{ABO13}, if equation \eqref{ABOWBO1} holds, then the probability density function of $Y_{i1}|\mathbf{Y}_{i2}=\mathbf{y}_{i2}$ can be given by
\begin{equation}\label{Oromo20}
f\left(y_{i1}|\mathbf{y}_{i2}\right) = \mathcal{N}\left(y_{i1}|\eta^{c}_{1}, \Sigma^{c}_{11}\right)\frac{\Phi\left(\boldsymbol\lambda'\boldsymbol\Sigma^{-1/2}\left(\mathbf{y}_{i}-\boldsymbol\eta \right)\right)}{\Phi\left(\boldsymbol\tau'\boldsymbol\Sigma_{22}^{-1/2}\left(\mathbf{y}_{i2}-\boldsymbol\eta_{2} \right)\right)},
\end{equation}
where
\begin{equation*}
\eta^{c}_{1}=\eta_{1}+\boldsymbol\Sigma_{12}\boldsymbol\Sigma^{-1}_{22}\left(\mathbf{y}_{i2}- \boldsymbol\eta_{2}\right) \quad\text{and}\quad\Sigma^{c}_{11}=\Sigma_{11}-\boldsymbol\Sigma_{12}\boldsymbol\Sigma^{-1}_{22}\boldsymbol\Sigma_{21}.
\end{equation*}
\begin{proposition}
Using equation \eqref{Oromo20}, the expected value of $Y_{i1}|\mathbf{Y}_{i2}=\mathbf{y}_{i2}$  can be given by 
\begin{equation}\label{Oromo16}
E\left[Y_{i1}|\mathbf{Y}_{i2}=\mathbf{y}_{i2}\right] = \eta^{c}_{1}+\frac{\Sigma^{c}_{11}\nu_{1}}{\sqrt{1+\nu_{1}\Sigma^{c}_{11}\nu_{1}}}\frac{\phi\left(\boldsymbol\tau'\boldsymbol\Sigma_{22}^{-1/2}\left(\mathbf{y}_{i2}-\boldsymbol\eta_{2} \right)\right)}{\Phi\left(\boldsymbol\tau'\boldsymbol\Sigma_{22}^{-1/2}\left(\mathbf{y}_{i2}-\boldsymbol\eta_{2} \right)\right)}.
\end{equation}
\end{proposition}
\begin{proof}
Let $\kappa\left(\mathbf{y}_{i2}\right) =  \Phi\left(\boldsymbol\tau'\boldsymbol\Sigma_{22}^{-1/2}\left(\mathbf{y}_{i2}-\boldsymbol\eta_{2} \right)\right)$ and $B = \boldsymbol\nu_{2}' \left(\mathbf{y}_{i2} - \boldsymbol\eta_{2}\right) - \nu_{1}\eta_{1}$. The expected value of $Y_{i1}|\mathbf{Y}_{i2}=\mathbf{y}_{i2}$ can be given by 
\begin{eqnarray*}
E\left[Y_{i1}|\mathbf{Y}_{i2}=\mathbf{y}_{i2}\right] &=& \displaystyle\int_{\Bbb R}y_{i1}f\left(y_{i1}\mid\mathbf{y}_{i2}\right)dy_{i1},\notag\\&=& \frac{1}{\kappa\left(\mathbf{y}_{i2}\right)}\displaystyle\int_{\Bbb R}y_{i1}\mathcal{N}\left(y_{i1}\mid\eta^{c}_{1}, \Sigma^{c}_{11}\right)\Phi\left(\boldsymbol\lambda'\boldsymbol\Sigma^{-1/2}\left(\mathbf{y}_{i}-\boldsymbol\eta \right)\right)dy_{i1},
\\&=& \frac{1}{\kappa\left(\mathbf{y}_{i2}\right)}\displaystyle\int_{\Bbb R}\int_{0}^{\infty}y_{i1}\mathcal{N}\left(y_{i1}\mid\eta^{c}_{1}, \Sigma^{c}_{11}\right)\mathcal{N}\left(x\mid \nu_{1}y_{i1}+B, 1\right)dxdy_{i1},
\\&=& \frac{1}{\kappa\left(\mathbf{y}_{i2}\right)}\int_{0}^{\infty}\mathcal{N}\left(x\mid B +  \nu_{1}\eta^{c}_{1}, 1 + \nu_{1}\Sigma^{c}_{11}\nu_{1}\right)E[Y_{i1}|X=x]dx,
\end{eqnarray*} 
where 
\begin{eqnarray*}
Y_{i1}|X=x &\sim& \mathcal{N}\left(x\mid \eta^{c}_{1} + \Lambda\nu_{1}\left(x - B - \nu_{1}\eta^{c}_{1}\right),   \Lambda\right) \quad\text{with} \quad\Lambda = \left(\frac{1}{\Sigma^{c}_{11}} + \nu^{2}_{1}\right)^{-1}.
\end{eqnarray*} 
The result in equation \eqref{Oromo16} follows from well known properties of the truncated normal distribution and inverse matrix adjustment formula. 
\end{proof}
Using equation \eqref{Oromo16}, we can obtain the point estimator of  $Y_{i1}$  by
\begin{eqnarray*}
	E\left[Y_{i1}|\mathbf{Y}_{i2}\right] &=&E\left[E\left[Y_{i1}|\mathbf{Y}_{i2}, Z_{i}\right]|\mathbf{Y}_{i2}\right],\\
	&=&\sum_{k=1}^{K}P\left(Z_{i} =k|\mathbf{Y}_{i2}\right) E\left[Y_{i1}|\mathbf{Y}_{i2}, Z_{i} =k\right]. 
\end{eqnarray*}
In this framework, the latent variable $Z_{i}$ represents the underlying tissue classes.
The weight $P\left(Z_{i} =k|\mathbf{Y}_{i2}\right)$ can be computed using Bayes' theorem. The expected value  $ E\left[Y_{i1}|\mathbf{Y}_{i2}, Z_{i} =k\right]$ is obtained from equation \eqref{Oromo16} by indexing the parameters with $k$.  
\subsection{Model validation}
The main focus of this work is to investigate CT image predictive quality of SGMM and GMM by partitioning the data into two major tissue types. It is also aimed at comparing the estimation quality of the partitioning approach with GMM, HMM, and MRF on the bone tissues. 

We used leave-one-out cross-validation method to compare the predictive quality of the models. That is, we keep one head to validate the models and use the remaining heads for training the models. Using the threshold CT image intensity, CT image in a validation dataset is partitioned into non-bone and bone tissues.  Let ${Y}^{\left(t\right)}_{i}$  and $\hat{Y}^{\left(t\right)}_{i}$ be CT image and its corresponding estimated CT image intensities in a given tissue region $t$.  One can use the square and the absolute loss functions to assess the estimation cost. They can be given by
\begin{eqnarray*}
	\left(\hat{Y}^{\left(t\right)}_{i}- {Y}^{\left(t\right)}_{i}\right)^2 \quad \text{and} \quad |\hat{Y}^{\left(t\right)}_{i}- {Y}^{\left(t\right)}_{i}|,
\end{eqnarray*} 
 where $t = 1, 2.$ Since the mean square error heavily weights the outliers, the mean absolute error (MAE) can be employed as a main tool to evaluate the estimation performance of the models.   The  MAE can be given by 
\begin{eqnarray*}
	MAE_{t}=\frac{1}{n_{t}}\displaystyle\sum_{i=1}^{n_{t}}|\hat{Y}^{\left(t\right)}_{i}- {Y}^{\left(t\right)}_{i}|,
\end{eqnarray*}
where $n_{t}$ is the number of voxels in partition $t$.

The peak signal-to-noise ratio (PSNR) can also be used to quantify the overall quality of the estimation. It takes square loss function into account through the mean squared error. The PSNR may be given by
\begin{eqnarray*}
	PSNR=10\log_{10} \left(\frac{nM^{2}}{\displaystyle\sum_{i=1}^{n}\left( \hat{Y}_{i}- {Y}_{i}\right)^{2}}\right),  
\end{eqnarray*}
where  $Y_{i}$ and $\hat{Y}_{i}$ represent CT image and the estimated CT image intensities at voxel $i$, respectively, and $M$ is the  maximal intensity in CT image. We exploited patient-wise leave-one-out cross-validation and  mean absolute error to compare the estimation quality of the models. One can also use peak signal-to-noise ratio to compare the estimation performance of the models. The better model has lower MAE and higher PSNR. Since MAE and PSNR are crude estimation quality measures, we utilised smoothed residual and absolute residual plots  to further evaluate the estimation quality of the models through the tissues of the heads. A moving average over non-overlapping windows in CT image intensities can be used as a main tool to investigate and identify the model that outperforms in bone tissue estimation.  Over the non-overlapping windows on ${Y}_{i}$ with a window size of 20 HU, the averages for ${Y}_{i}$,  $\hat{Y}_{i}- {Y}_{i}$, and $|\hat{Y}_{i}- {Y}_{i}|$ over the windows can be computed to obtain the smoothed residual and absolute residual plots.  

In addition to the above model performance evaluation methods, we assessed the performance of the models using Bland-Altman plot \citep{ABO32, ABO33}, which is a graphical method to compare two measurements by plotting the differences between the two measurements $\hat{Y}_{i}- {Y}_{i}$ against their averages $( \hat{Y}_{i} + {Y}_{i})/2$ .  If one of the two  measurements is a reference measurement, then the differences can be plotted against the reference measurement \citep{ABO34} and that coincides with the smoothed residual plot via cross-validation. The main advantage of the Bland-Altman plot is that it reveals a relationship between the differences and the magnitude of the measurements, to examine possible systematic bias, and outliers.

We can also complement the MAE based performance evaluation of the methods by Wilcoxon signed-rank test,  which is a non-parametric alternative to the paired Student's t-test. To compare SGMM to every one of the remaining methods, the paired method-method differences in MAEs were examined by using the one-sided Wilcoxon signed rank test, with the null hypothesis that SGMM is the same as one of the other methods with respect to the MAE values and the research hypothesis that SGMM is better than the method being compared.
\section{Results}
A CT image intensity 100HU was utilized as a delimiting value  during CT image estimation. The optimal parameter estimates were received for $K=8$ in GMM and $K=5$ for both HMM and MRF. In the case of SGMM and GMM*, we have received the optimal parameter estimates for $K=6$ for both major tissue types. Table~\ref{tab:Bone_region} demonstrates a summary of mean absolute errors for the bone tissues  and p-values of the Wilcoxon signed-rank test.

\begin{table}[H]
\tbl{Mean absolute errors and p-values  for bone tissues }
{\begin{tabular}{c|cccccc} \toprule
 & \multicolumn{2}{l}{Model} \\
 \cmidrule{2-6}
 Head &SGMM  &GMM$^{*}$ &HMM  &MRF  &GMM \\ \midrule
  1&316.25 &315.16& 324.94& 307.95& 314.41 \\ 
  2&349.52& 348.05&  360.03&  328.89& 365.12\\ 
  3&303.58& 301.78&  331.16&  322.48&  328.01\\ 
  4&272.21 &269.33&  296.04&  280.63&  292.78\\ 
  5&350.41&349.28&  366.13&  357.22&  359.56\\\bottomrule
Average&318.39& 316.72& 335.66 &319.43 &331.98\\\midrule
p-value&-& 0.985&0.030  &0.500 &0.053\\\bottomrule
\end{tabular}}
\label{tab:Bone_region} 
\end{table}

The rows of the table represent validation datasets. The table presents mean absolute errors received for the models.  In terms of MAE, it is apparent that for each head the partitioning approach (SGMM and GMM*) and MRF had better performance than HMM and GMM. Moreover, the differences of the average MAEs show that the partitioning approach and MRF achieved better results than HMM and GMM. For instance, SGMM  outperforms HMM with the method-to-method difference of average MAEs -17.27 and the standard deviation of the  method-to-method differences of MAEs 8.23. The p-values of Wilcoxon signed-rank test show that SGMM outperforms HMM  and GMM  significantly on bone tissues.

 The results in Table~\ref{tab:Dense_Bone_region}  show that the partitioning approach has noticeable outperformance on dense bone tissues (approximately with CT image intensities greater than 900HU according to  Washington and Leaver \citep{ABO20}) as compared to the remaining models.  Even though MRF is computationally expensive, it had better performance than HMM and GMM on dense bone tissues. The  p-values of Wilcoxon signed-rank test reveal that SGMM has  significantly better performance than the remaining methods except for  GMM$^{*}$ on dense bone tissues.

\begin{table}[H]
	\tbl{Mean absolute errors  and p-values for dense bone tissues}
	{\begin{tabular}{c|cccccc} \toprule
			& \multicolumn{2}{l}{Model} \\
			\cmidrule{2-6}
			Head &SGMM  &GMM$^{*}$ &HMM  &MRF  &GMM \\ \midrule
			1   &365.68 &361.27&418.99&401.43 &407.85 \\ 
			2	&406.09 &403.65&458.23&394.89 &492.96\\ 
			3	&331.85 & 328.44 &407.01& 404.80 & 386.72\\ 
			4	&236.89 & 232.15 & 290.42& 295.46 & 290.77\\ 
			5   &436.42 & 433.84 & 509.20& 488.10 & 505.49\\ \bottomrule
			Average&355.39&351.87& 416.77&396.94&416.76\\	\midrule
			p-value&-&0.985 & 0.030 & 0.053 &0.030\\\bottomrule
		\end{tabular}}
		\label{tab:Dense_Bone_region}
	\end{table}
	
Table~\ref{tab:Non_Bone_region} demonstrates the estimation quality of the models on non-bone tissues. The best  result was received for HMM. However, there is already  a good contrast between soft tissues and air on MR images. The remaining models had similar behavior on non-bone tissues.

\begin{table}[H]
	\tbl{Mean absolute errors  and p-values for non-bone tissues}
	{\begin{tabular}{c|cccccc} \toprule
			& \multicolumn{2}{l}{Model} \\
			\cmidrule{2-6}
			Head &SGMM  &GMM$^{*}$ &HMM  &MRF  &GMM \\ \midrule
		1	& 111.96& 112.67&97.75&106.15& 114.10\\ 
		2	&116.56& 117.53&99.69& 103.60&116.12\\ 
		3	&124.65& 125.46&94.93&117.31& 122.00\\ 
		4	&116.61&117.19&101.54& 112.52&111.76\\ 
		5	&122.08& 122.42&98.87&126.51&118.19\\\bottomrule
		Average	&118.37& 119.05& 98.55&113.22& 116.43\\\midrule
		p-value&-& 0.030&0.985  & 0.947& 0.947\\\bottomrule
		\end{tabular}}
		\label{tab:Non_Bone_region}
	\end{table}
	
Table~\ref{tab:Overall} presents the overall summary of CT image estimation quality. In comparison to the other models, we received better result for HMM. The reason is that HMM outperformed the other models on non-bone tissues. However, this is not the main interest in this work. 

\begin{table}[H]
	\tbl{Combined mean absolute errors  and p-values}
	{\begin{tabular}{c|cccccc} \toprule
			& \multicolumn{2}{l}{Model} \\
			\cmidrule{2-6}
			Head &SGMM  &GMM$^{*}$ &HMM  &MRF  &GMM \\ \midrule
		1	& 144.15 &144.58&133.56&137.95& 145.67\\ 
		2	& 151.44 &152.04& 138.68 & 137.34&153.40 \\ 
		3	& 161.00 &161.28&142.92 & 158.99 & 163.85 \\ 
		4	& 146.71 &146.62& 139.17 &  145.05 &146.78\\ 
		5	&160.53 & 160.62&  143.89 &165.36&158.85\\\bottomrule
		Average  &152.77& 153.03&  139.65  &148.94 & 153.71\\\midrule
		p-value&-& 0.069& 0.985 &0.911 & 0.140\\\bottomrule
		\end{tabular}}
	  \label{tab:Overall} 
	\end{table}
	
Table~\ref{tab:PSNR} demonstrates the prediction accuracy of the models in terms of PSNR. The results show that the models had similar behavior. On the bone tissues, PSNR has shown that the models have similar estimation performance.

\begin{table}[H]
	\tbl{Evaluation of the methods based on PSNR and p-values}
	{\begin{tabular}{c|cccccc} \toprule
			& \multicolumn{2}{l}{Model} \\
			\cmidrule{2-6}
			Head &SGMM  &GMM$^{*}$ &HMM  &MRF  &GMM \\ \midrule
		1	&19.92&19.90&20.29& 20.30& 20.29\\ 
		2	&19.56 &19.50& 20.23 & 20.40&20.15 \\ 
		3	&19.75 &19.71&20.13 & 20.11 & 19.69 \\ 
		4	&19.76 &19.76& 20.43 & 20.46 &20.69\\ 
		5	&19.59 &19.58& 19.78 &19.63&19.66\\\bottomrule
		Average   &19.72&19.69&  20.17  &20.18 & 20.09\\\midrule
		p-value&-&0.978 & 0.030  &0.030 &0.053\\\bottomrule
		\end{tabular}}
		\label{tab:PSNR} 
\end{table}

On the basis of average, mean absolute error was utilized to compare CT image prediction accuracy of the models. Smoothed absolute residual plots were also employed to assess the estimation quality of the models. In comparison to MAE, smoothed absolute residual plots are powerful to evaluate the estimation performance of the models through the tissues of the heads.  Figure~\ref{Or3} presents smoothed absolute residual plots for the models. The absolute residuals were averaged over non-overlapping windows in CT image intensities with window size 20HU. 

\begin{figure}[H]
	\centering
	\includegraphics[height=8cm, width=13cm]{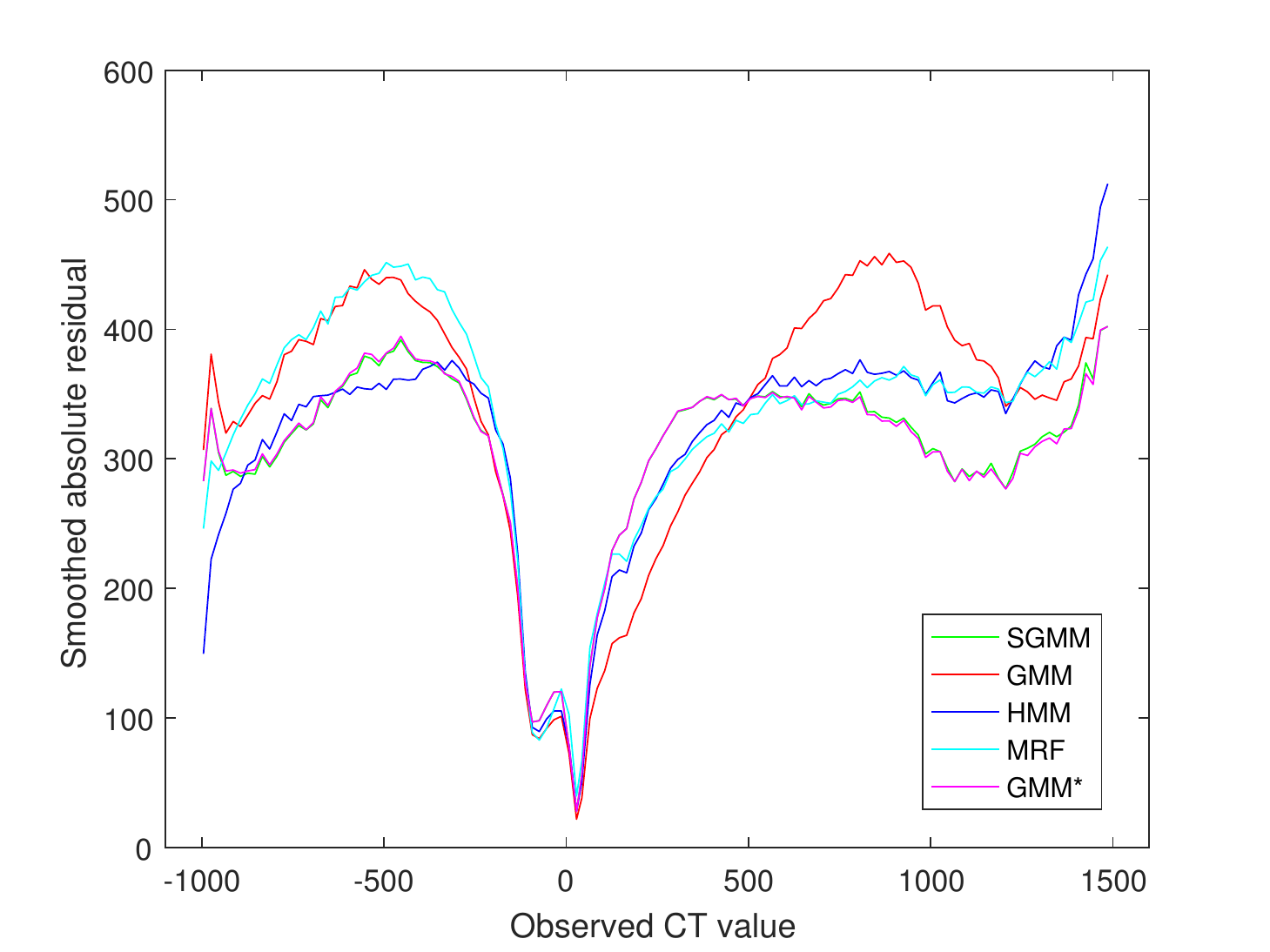}
	\caption{Smoothed absolute residual plot for the five patients}
	\label{Or3}
\end{figure}

It is clear from Figure \ref{Or3} that none of the models outperformed throughout the tissues of the head. However, SGMM and GMM$^{*}$ outperformed the other models on dense bone tissues. Figure \ref{Or4} shows smoothed residual plot. The deviations of observed CT image intensities  from the estimated CT image intensities  were exploited to obtain average residuals over non-overlapping windows in CT image intensities with window size 20HU. It is evident from the plot that bone tissues  were underestimated. However, the partitioning approach has improved underestimation of bone tissues. 

\begin{figure}[H]
	\centering
	\includegraphics[height=8cm, width=13cm]{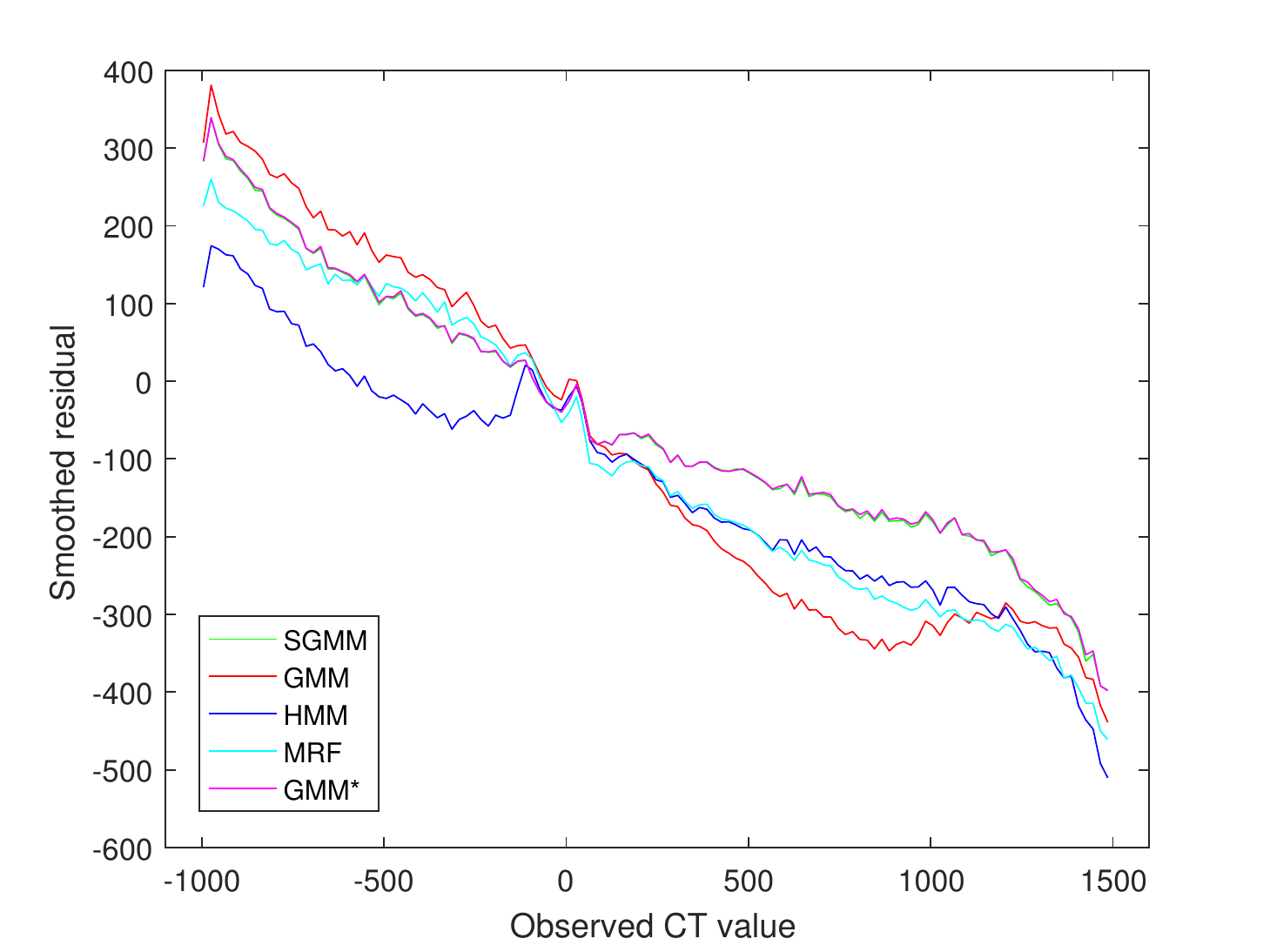}
	\caption{Smoothed residual plot for the five patients}
	\label{Or4}
\end{figure} 
The Bland-Altman plots of the methods are shown in Figure \ref{BlandAltman}. It shows the bias of the methods in estimating the CT images. The bias is higher for higher CT image values. However, the partitioning approach has improved the bias in bone tissue estimation as compared to the remaining methods. The Bland-Altman plots also show higher variability on lower CT image values.  In comparison to \ref{Or4}, we observe that the higher variability on lower CT image values  in Figure \ref{BlandAltman} is due the estimated CT image values. Notice, however, that this work is mainly concerned with improving bone tissue estimation.

\begin{figure}[H]
	\centering
	\includegraphics[height=8cm, width=13cm]{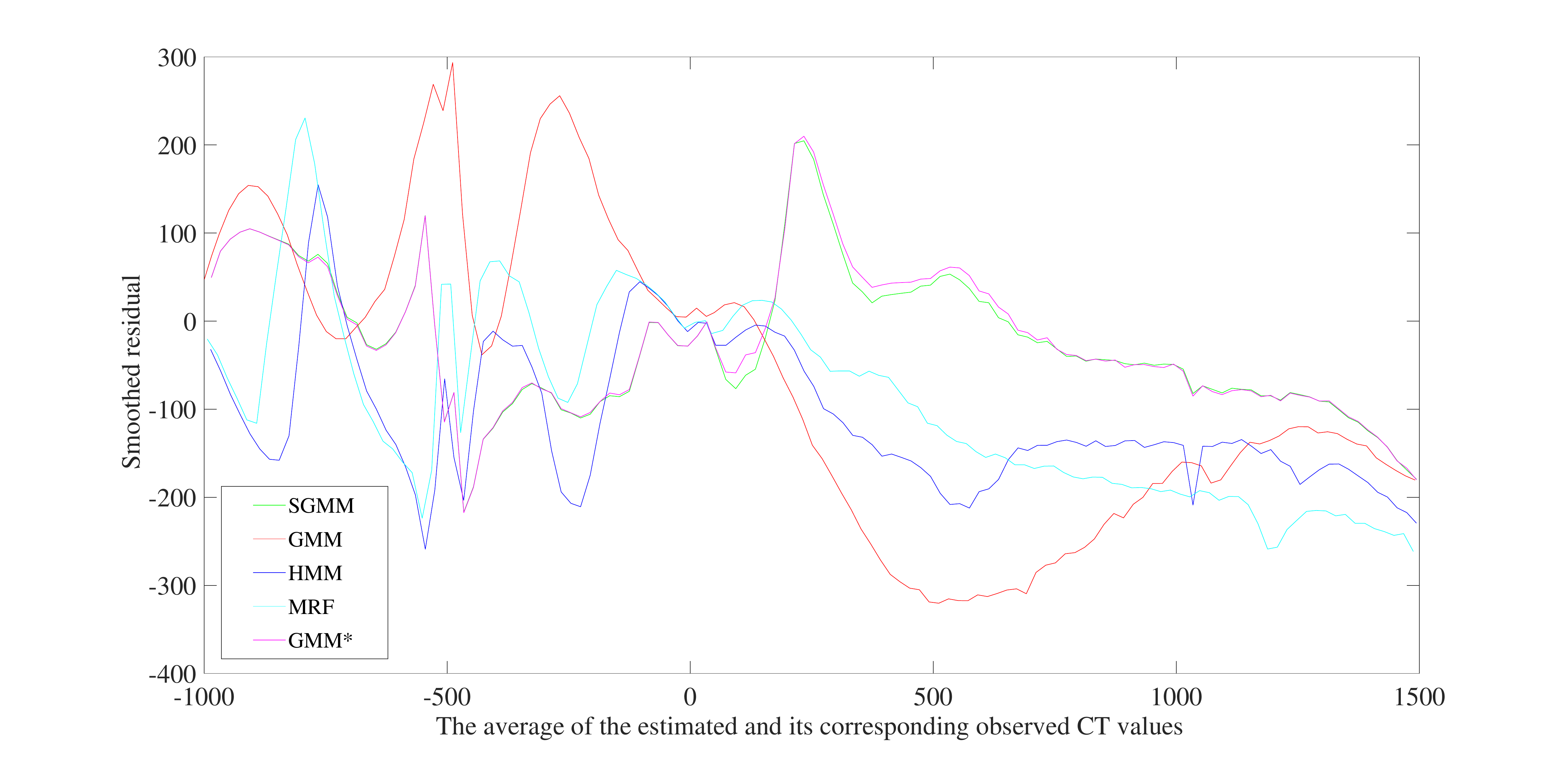}
	\caption{Bland-Altman plot for the five patients}
	\label{BlandAltman}
\end{figure}

The partitioning approach had better performance in tracing bone tissues. Bone tissues appear as bright white on CT images, see Figure \ref{Or45}.  The figure shows slices of CT image and its corresponding predicted slices for a representative patient. The first row in the figure where the images are marked by red colored box clearly shows that the partitioning approaches are better in identifying bone tissues. 

\begin{figure}[H]
	\vspace*{-0.33cm}
	\centering
	\hspace*{-2.3cm}
	\includegraphics[height=9cm, width=18cm]{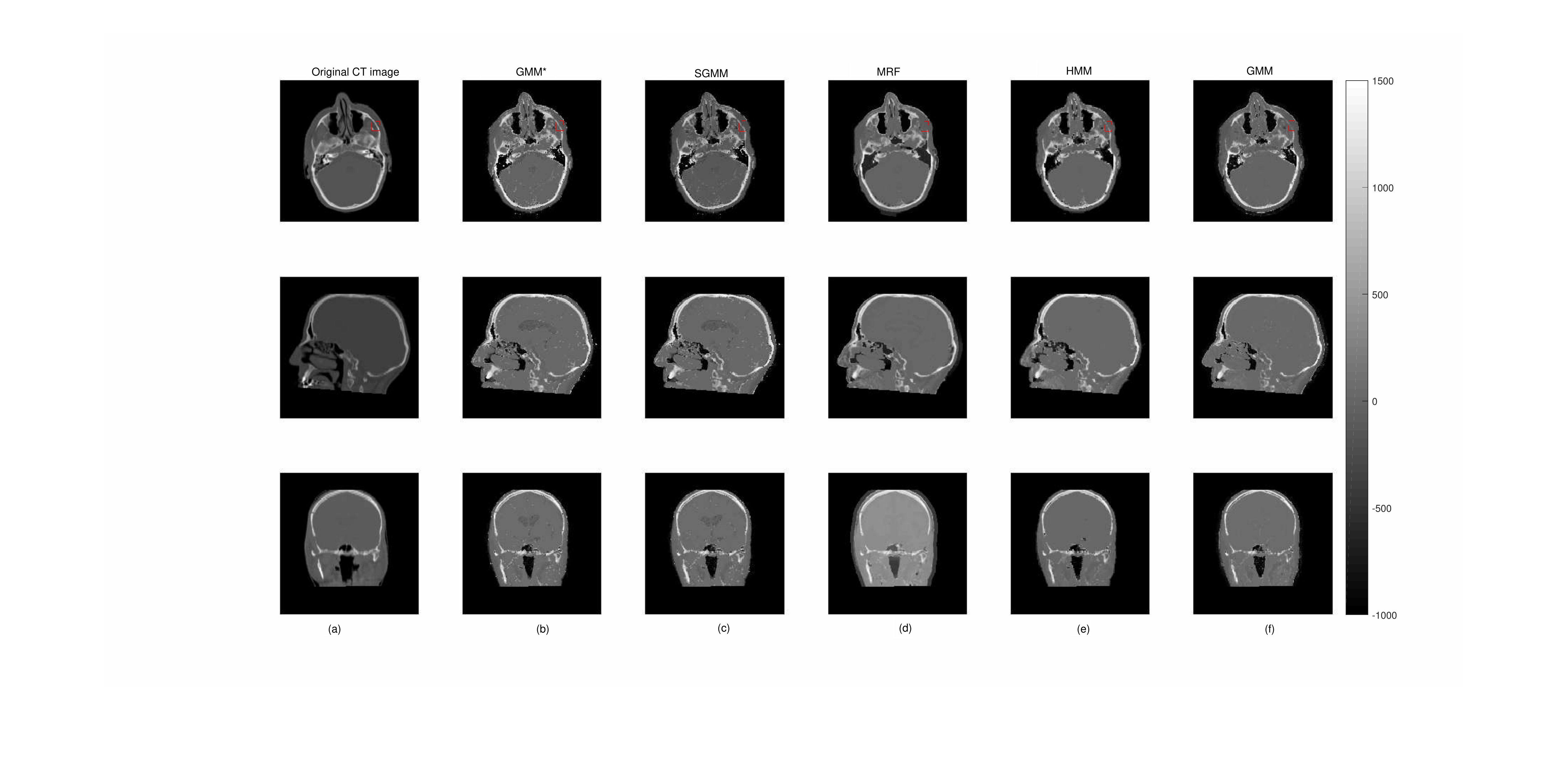}
	\vspace*{-1.45cm}
	\small\caption{The first column (a) presents the original CT image slices and the remaining columns (b-f) show the corresponding  predicted slices of CT image.}
	\label{Or45}
\end{figure}
We presented the prediction errors in Figure \ref{Or45}, which corresponds to the predicted slices in Figure \ref{Or65}. It can be seen that the images of the prediction errors corresponding to the bone tissues appear darker for the partitioning approach:  GMM$^{*}$ and SGMM.

\begin{figure}[H]
	\vspace*{-0.1cm}
	\centering
	\hspace*{-2.3cm}
	\includegraphics[height=9cm, width=18cm]{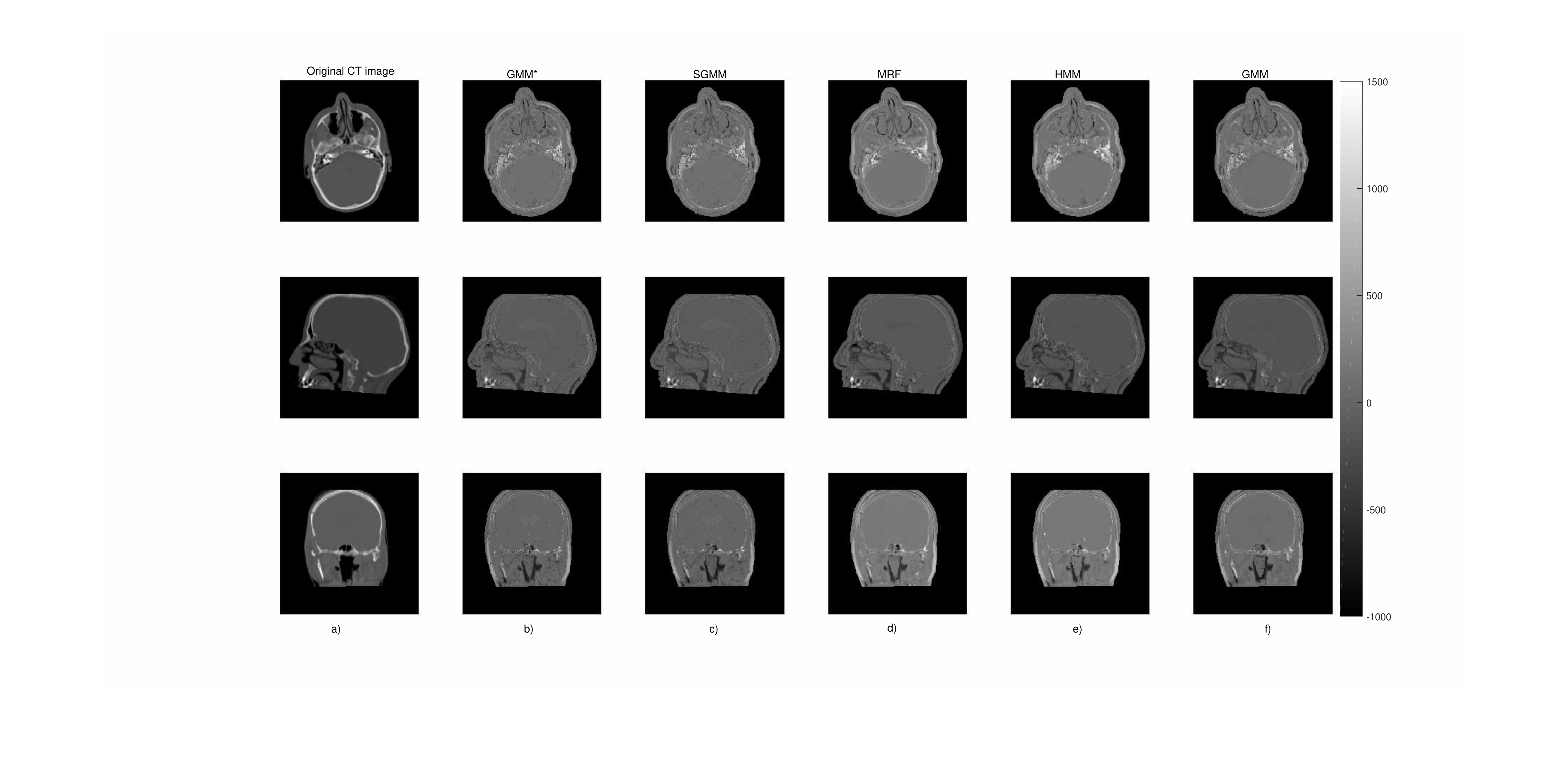}
	\vspace*{-1.45cm}
	\small\caption{The first column (a) presents the original CT image slices and the remaining columns (b-f) show the prediction errors for each model.}
	\label{Or65}
\end{figure}
\section{Discussion}
Statistical model for voxel-wise CT image estimation from MR images have been presented and evaluated using cross-validation on five datasets. Kuljus et al. \citep{ABO14} and Johansson et al. \citep{ABO10} have used  voxel-wise estimation approach to study CT image estimations. Existing works suggest that the estimation quality on the bone tissues is poor. This study was aimed at probing the estimation of CT images by partitioning the data into two major tissue types: non-bone and bone tissues. Specifically, the focus of the study was to obtain improved bone tissue estimations. We proposed SGMM to relax the  distributional assumption utilised by  Kuljus et al. \citep{ABO14} and Johansson et al. \citep{ABO10}. The model was motivated to take the asymmetrical distribution that could arise from the partitioning of the data into account. We also used GMM in addition to SGMM to examine the effect of the distributional nature of the partitioned data on the estimation process. The study was also aimed at comparing CT image estimation performance of SGMM, HMM, GMM, and  MRF on the bone tissues. In MRF and HMM, spatial dependence between neighboring voxels has been taken into account. 

EM-algorithm was developed for SGMM parameter estimation. EM-gradient algorithm was used to estimate the parameters of MRF while EM-algorithm was utilized to estimate the parameters of GMM and HMM. For MRF, the updates of the parameters at M-step of EM-algorithm were difficult to obtain in an explicit form and thereby, a gradient based optimization was utilized during parameter estimation process. Moreover, Gibbs sampling was used at E-step of the algorithm. Thus, the estimation process was expensive in MRF. Unlike the other models, log-likelihood function in MRF involves Gibbs field and it is not computable. This means that log-likelihood based selection of optimal parameter estimate is not feasible analytically. Consequently, mean squared error was employed in selecting the optimal parameter estimates of the models. 

Table~\ref{tab:Overall} demonstrated that HMM outperformed the other models. This was essentially due to its best estimation performance on non-bone tissues, see Table~ \ref{tab:Non_Bone_region}.  According to Karlsson et al. \cite{ABO12}, the key task in  CT image estimation from MR images is to obtain an improved bone tissue estimation.   Table~\ref{tab:Bone_region} revealed that GMM* (GMM trained on each major tissue type), SGMM, and MRF had better prediction accuracy than HMM and GMM on the bone tissues. In addition, Table~\ref{tab:Dense_Bone_region} shows that GMM* and SGMM had noticeable outperformance on dense bone tissues.  Based on mean absolute error, HMM and GMM perform similarly on the bone tissues. 

To provide statistical evidence for the CT image estimation performance of SGMM in comparison with the remaining methods, the  Wilcoxon signed-rank tests have been carried out. The tests has shown that SGMM has better performance on bone tissues than HMM and GMM. In addition, the Wilcoxon signed-rank tests provided an evidence that SGMM has significant outperformance than MRF, HMM, and GMM on dense bone tissues. 

The skewness assumption allowed to recognize skewness in the partitions of the data.  The estimates of the skewness parameters demonstrated that the partitions have skewness property and the nature of skewness was dependent on the subtissue types. The estimates of the skewness parameters ranged from -1.54 to 3.03. In this particular application,  however, the skewness assumption did not improve CT image estimation quality as compared to the symmetric assumption in GMM*, see Table~ \ref{tab:Bone_region}--\ref{tab:PSNR}.

Figure~\ref{Or3} shows that the models performed better on soft tissues, that is tissues having CT image intensity closer to 0HU. On the other hand, the models had weaker prediction accuracy on the two extremes that is on air and bone tissues. This pattern of residual plots have been observed in \citep{ABO10, ABO14}. Moreover, the figure shows that none of the models outperformed throughout the tissues of the head. Table~\ref{tab:Bone_region}--\ref{tab:Overall} demonstrate that the predictive accuracy of the models was dependent on the heads. That is the results were not uniform over the heads. This might have a problem in real applications and  needs a further investigation. 

The bias of the methods in estimation of the CT images was manifested in the smoothed residual and Bland-Altman plots  in Figures \ref{Or4} and \ref{BlandAltman}. Though the bias is higher on bone tissues, it was improved by our approach as compared to the remaining methods. This was the main concern of this work. The Bland-Altman plots also demonstrated higher variability of the estimation on lower CT image values. 
\section{Conclusions}
In this study, we examined CT image estimation by partitioning the data into two major tissue types. This partitioning  approach is an efficient way to get a good quality CT image substitute with improved estimation of bone tissues. Moreover, the SGMM and the developed algorithm to estimate its parameters is general and can be applied to other applications.
\section*{Acknowledgments}
 \noindent  This work is supported by the Swedish Research Council grant (Reg. No. 340-2013-5342). The authors thank Adam Johansson and Thomas Asklund for providing us data and David Bolin for providing us Mathlab code for MRF. Moreover, the authors thank Kristi Kuljus for HMM results. The computations were performed on resources provided by the Swedish National Infrastructure for Computing (SNIC) at High Performance Computing Center North (HPC2N).

\end{document}